\documentclass[12pt]{article}
\usepackage[T1]{fontenc}
\usepackage{mathpazo}
\usepackage{amssymb,amsmath,amsthm,bm,mathtools,mathrsfs,mathalfa,mathabx, mathtools}
\usepackage{enumitem}
\usepackage[margin=1in]{geometry}
\usepackage{comment}
\usepackage{float}
\usepackage{natbib}
\usepackage{booktabs}
\usepackage{setspace}
\usepackage{multirow}
\usepackage[flushleft]{threeparttable}
\usepackage{tikz}
\usepackage{caption}
\usepackage{xcolor}
\usepackage{bbm}
\usepackage{graphicx}

\usepackage{hyperref}
\hypersetup{
	colorlinks = true,
	linkcolor = blue,
	citecolor = blue,
}

\newtheorem{theorem}{Theorem}

\newtheorem{assumption}{Assumption}

\newtheorem{corollary}{Corollary}

\newtheorem{manualtheoreminner}{Assumption}
\newenvironment{manualtheorem}[1]{%
  \renewcommand\themanualtheoreminner{#1}%
  \manualtheoreminner
}{\endmanualtheoreminner}

\makeatletter
\def\@xfootnote[#1]{%
  \protected@xdef\@thefnmark{#1}%
  \@footnotemark\@footnotetext}
\makeatother

\makeatletter
\def\titlefootnote{\ifx\protect\@typeset@protect\expandafter\footnote\else\expandafter\@gobble\fi}
\makeatother

\begin{document}
    
\begin{center}
	{\Large \textbf{Global Representation of the Conditional LATE model: \\A Separability Result}}

	\vspace{2ex}
    Yu-Chang Chen and Haitian Xie\footnote[$\asterisk$]{Alphabetical ordering; both authors contributed equally to this work. For helpful comments and suggestions, we would like to thank Yixiao Sun, Kaspar W\"uthrich, and two anonymous referees.}
    
    \vspace{1ex}
	Department of Economics, UC San Diego\footnote[\textdagger]{Email: yuc391@ucsd.edu, hax082@ucsd.edu. Correspondence to: Department of Economics, University of California, San Diego, 9500 Gilman Drive, La Jolla, CA 92093-0508.}

	\vspace{1ex}
	\today
\end{center}

%\tableofcontents

\vspace{5ex}

\begin{abstract}
    
    This paper studies the latent index representation of the conditional LATE model, making explicit the role of covariates in treatment selection. We find that if the directions of the monotonicity condition are the same across all values of the conditioning covariate, which is often assumed in the literature, then the treatment choice equation has to satisfy a separability condition between the instrument and the covariate. This global representation result establishes testable restrictions imposed on the way covariates enter the treatment choice equation. We later extend the representation theorem to incorporate multiple ordered levels of treatment.

	\vspace{5ex}

    %\noindent \textbf{Word Count: } 4673

	%\vspace{2ex}

	\noindent \textbf{JEL Classification: } C21, C50

	\vspace{2ex}

	\noindent \textbf{Keywords: } Local instrumental variables, Latent index representation, Conditioning covariates, Monotonicity, Separability.

\end{abstract}

\section{Introduction}

Self-selection into treatment is a common challenge in causal inference. One approach, pioneered by \cite{heckman1976common}, is to impose a model on the selection process. Another approach is to invoke the assumptions of \cite{imbens1994identification} and use instrumental variables to identify the local average treatment effect (LATE). \cite{vytlacil2002independence} shows that the two approaches are equivalent, even though the LATE approach does not provide an explicit model of the selection process. Specifically, \cite{vytlacil2002independence} finds that the monotonicity and independence conditions imposed in \cite{imbens1994identification} together imply a nonparametric binary choice model, in which the instrument and the unobserved heterogeneity are additively separable in the latent index. When conditioning covariates are included, say, for instrument validity, the selection model representation can be established on a given value of covariates. Namely, holding fixed the value of covariates, imposing a (nonparametric) selection model is no stronger than imposing the LATE assumptions.

However, in most empirical settings, a fully nonparametric analysis, conditioning on each value of the covariates, is prohibitively data-demanding. It is typical to pool observations with different characteristics and to incorporate covariates in the selection model (for example, \citealp{carneiro2011estimating} and \citealp{cornelissen2018benefits}). %The estimation of treatment effects is carried over the whole sample instead of nonparametric subgroup analysis. 
These empirical works conduct global analysis that explicitly models covariates while
the theoretical analysis of \cite{vytlacil2002independence} is local in the sense that covariates are fixed at a constant level.\footnote{The terminology ``local'' also appears in the work by \cite{dahl2020nevertoolate} with a different meaning. They consider the weakening of the LATE assumption based on the outcome distributions rather than the covariates. } This paper aims at filling this gap.\footnote{In a recent paper, \cite{kline2019heckits} compares the selection model and the LATE approach when covariates are present. While they have established that, in the absence of covariates, the selection model and the LATE approach will yield numerically identical estimates, they also point out that their equivalence result does not hold when the covariates are introduced at least in the way covariates are usually modeled in empirical works. This paper addresses the same issue by characterizing the set of selection models that are equivalent to the LATE model when covariates are present.}

%One important case is when the instrument can either incentivize or disincentivize the take-up of the treatment depending on the value of the covariate. Put differently, for some value of the covariate, there could be no defiers but only compliers, while for other values, there could be no compliers but only defiers. This case does not contradict the equivalence theorem but violates the conditional local average treatment effect (CLATE) model assumptions as typically described in \cite{abadie2003semiparametric}, and \cite{frolich2007nonparametric}, \cite{hong2010semiparametric}. 

Our result extends the representation in \cite{vytlacil2002independence} to settings where the covariates are not held fixed. We show that the conditional LATE (CLATE) model \citep{abadie2003semiparametric} has a threshold-crossing representation in which the instruments are separated from the covariates in the latent index. Loosely speaking, the separability is a result of the monotonicity condition in the CLATE model that requires the instruments to affect the potential treatment status in the same direction for all individuals.\footnote{\cite{heckman2005structural} suggests renaming the monotonicity condition as uniformity because ``it is a condition across people than the shape of a function for a particular person.'' } In particular, the direction of the monotonicity condition is the same across all values of the covariates, and thus a latent index representation that separates out the covariates is appropriate in this case.  %Furthermore, in the representation we derived, the covariates essentially serve as determinants of the threshold, and individuals select into the treatment if and only if the values of the instruments exceed that threshold.

The separability result implies that it is possible to uniformly rank the instruments' values by the propensity score across the covariates' values. Such a ranking has two practical usages. First, it can be used as a single index to construct testable implications for the CLATE model. Second, we can use the ranking to relabel the values of the instruments so that an increase in the ranking never causes individuals to drop the treatment, regardless of their covariates values. Our representation result also implies that the techniques developed in the conditional LATE (CLATE) framework, such as the identification analysis in \cite{abadie2003semiparametric}, can be applied to selection models that impose separability, and vice versa. As a corollary of the representation theorem. We reformulate and refine the testable implications of \cite{heckman2005structural} for the marginal treatment effect framework.

%The separability result implies that it is possible to uniformly rank the instrument values by the propensity score across the covariates values. It is under this rank-invariance property that the set of available compliance types (always-takers, compliers, never-takers, defiers) remains unchanged when the covariates vary, and in particular, the notion of complier, in the conditional setting, bears the commonly understood meaning as in \cite{imbens1994identification}. Our representation result also implies that the techniques developed in the conditional LATE (CLATE) framework, such as the identification analysis in \cite{abadie2003semiparametric}, can be applied to selection models that impose separability, and vice versa. As a corollary of the representation theorem, \textcolor{red}{we reformulate and refine the testable implications of \cite{heckman2005structural} for the marginal treatment effect framework.}
%changed%

The remaining of the paper is organized as follows. In section \ref{sec:main}, we present the global latent index representation of the LATE model, which is our main result. Section \ref{sec:implications} discusses some implications of the representation. Section \ref{sec:ordered} generalizes the result to the case of ordered-discrete choice selection models. The last section concludes.

\section{Representation Results} \label{sec:main}

We first introduce the CLATE model. Let the binary variable $D$ be the receipt of treatment so that $D=1$ denotes the treatment status, and $D=0$ denotes the untreated status. The potential outcomes under the treated and untreated status are denoted by $Y_1$ and $Y_0$, respectively. The actual outcome observed by the econometrician is $Y = DY_1 + (1-D)Y_0$. Let $X$ be a random vector containing variables that could potentially affect both the outcome and the treatment choice. The covariates are introduced in the model to make the validity of the instruments plausible. Denote $\mathcal{X}$ as the support of $X$. Let the random vector $Z$ be the collection of variables that affect the treatment choice $D$ but not the potential outcomes. These variables are referred to as instruments or excluded variables. Note that under this specification, $Z$ and $X$ are disjoint sets of variables. Denote $\mathcal{Z}$ as the support of $Z$. For each value $z$ of the instrument, let $D_z$ be the counterfactual treatment status if $Z$ were externally set to $z$. The realized treatment can be represented as $D = D_Z = \sum_{z \in \mathcal{Z}} \mathbf{1}\{Z=z\}D_z$.

%Assume that for any $(z,z') \in \mathcal{Z}^2$, and $x \in \mathcal{X}$, we have $P(D=1 \mid X) \in (0,1) $, and $P(D_z=1 \mid X=x) \ne P(D_{z'}=1 \mid X=x)$. This means that the instrument is relevant and the treatment choice is not deterministic given the instrument. 

To avoid measure-theoretic technicalities, we assume both $\mathcal{Z}$ and $\mathcal{X}$ are countable. We further assume that for any $x \in \mathcal{X}$, $P(D=1 \mid X=x) \in (0,1)$ and $P(D=1 \mid Z=z,X=x)$ is not a trivial function of $z$.%\footnote{The probability $P(D=1 \mid X=x)$ is bounded away from $0$ and $1$ as long as the population of compliers is not of measure zero.} 
This means that there exist both treated and untreated individuals, given each value of the covariates value. This assumption is also imposed in \cite{vytlacil2002independence}. The assumptions of the CLATE model are listed as follows.

\begin{assumption} [Conditional Independence] \label{ass:CI}
    $\left(  \{D_z: z \in \mathcal{Z} \},Y_1,Y_0\right) \perp Z \mid X $ .
\end{assumption}

\begin{assumption} [Monotonicity] \label{ass:monotonicity}
    For any $(z,z') \in \mathcal{Z}^2$, either 
    $$ P(D_z \geqslant D_{z'} \mid X=x) =1 \text{, for almost all } x$$
    or 
    $$ P(D_z \leqslant D_{z'} \mid X=x) =1 \text{, for almost all } x.$$\end{assumption}

Assumption \ref{ass:CI} requires the instrument to be ``as good as randomly assigned" conditional on the covariates. Assumption \ref{ass:monotonicity} is the monotonicity condition that is typically required in the LATE literature. Together, the two assumptions form the CLATE framework. Note that the exclusion restrictions of the instrument on the outcome is already embedded in the notation of the potential outcomes. 

We discuss the monotonicity condition in more detail. This condition is global as it requires the direction of monotonicity to be the same across different values of $x$. A weaker and conditional version of monotonicity would be to impose, for any $(z,z') \in \mathcal{Z}^2$, and for each $x$ locally, either
\begin{align} \label{eqn:lmc}
    P(D_z \geq D_{z'} \mid X=x) = 1 \text{, or } P(D_z \leq D_{z'} \mid X=x) = 1.
\end{align}
For any $x \in \mathcal{X}$, we can consider the individual with $P(D_z > D_{z'} \mid X=x) = 1$ as the complier and the individual with $P(D_z < D_{z'} \mid X=x) = 1$ as the defier. Then under the local monotonicity condition (\ref{eqn:lmc}), it is possible that for some $x$, there are compliers but no defier; while for other $x$, there are defiers but no complier. This notion of local monotonicity can be found, for example, in \cite{kolesar2013estimation} and \cite{sloczynski2020should}. However, it is important to have uniformity in the direction of monotonicity in order to obtain the separability result in the global representation. 

The main result of this paper is Theorem \ref{thm:representation}.

\begin{theorem} [Latent Index Representation of CLATE] \label{thm:representation}
    The following two representations are equivalent.
    \begin{enumerate} [label = (\roman*)]
        \item The CLATE model (Assumptions \ref{ass:CI} and \ref{ass:monotonicity}).
        \item There exist functions $m$ and $q$, and a random variable $U$ such that $(Y_1,Y_0,U) \perp Z \mid X$ and the treatment choice is determined by
        \begin{align} \label{eqn:index_crossing}
            D_z = \mathbf{1}\{ m(z) \geqslant q(X,U) \}  \; w.p.1.
        \end{align}
    \end{enumerate}
	Furthermore, if the conditional distribution of $U \mid X=x$ is absolute continuous for all $x \in \mathcal{X}$,\footnote{This assumption is typically imposed in the marginal treatment effect literature (for example, \citealp{heckman2005structural}) for the normalization of $U$.}  then there exist a function $q^*$ and a random variable $U^* \sim \text{Unif}[0,1]$ such that $U^* \perp (Z,X)$ and   
\begin{align} \label{eqn:representation_uncond}
            D_z = \mathbf{1}\{ m(z) \geqslant q^*(X,U^*) \}  \; w.p.1.
        \end{align}
\end{theorem}

This representation result achieves separability between the instrument and covariates in the treatment choice process. The function $m$ ranks the values of the instrument. Moreover, this ranking is invariant to changes in the covariates and is identified up to an increasing transformation. We further explain this ranking in the next section.

The form of Equation (\ref{eqn:index_crossing}) is to emphasize the separation between the instrument $Z$ and the covariates $X$. Alternatively, we can define $\tilde{U} = q(X,U)$, and write the selection equation as
\begin{align} \label{eqn:representation_without_X}
    D_z = \mathbf{1}\{ m(z) \geqslant \tilde{U} \} \; w.p.1, 
\end{align}
where $(Y_1,Y_0,\tilde{U}) \perp Z \mid X$. Representation (\ref{eqn:representation_without_X}) is used in the proof of Corollary \ref{cor:obs_implication}. Note that the separation between $Z$ and $X$ in representation (\ref{eqn:index_crossing}) and (\ref{eqn:representation_uncond}) holds inside the indicator function, and it does not necessarily imply that propensity score is additively separable in $Z$ and $X$. For example, consider the simple treatment selection equation $\mathbf{1}\{Z + X \geq U\}$, where $U \mid (Z,X) \sim N(0,1)$. In this case, the propensity score is equal to $\pi(z,x) = \Phi(z+x)$, with $\Phi$ being the distribution function of the standard normal distribution. This particular propensity is not additively separable between its two arguments.\footnote{That being said, non-separabilities of $Z$ and $X$ in the propensity score may lead to a contradiction to the monotonicity assumption. For example, this can happen if the marginal effect of increasing $Z$ is positive for some $(X,Z)=(x,z)$ but negative for another value $(X,Z)=(x',z)$. For example, if the propensity score is $\pi(z,x) = \Phi(zx)$, then the monotonicity assumption would be violated if $x$ can take both positive and negative values.}

The intuition of the Theorem is explained along with the following proof, where we make use of the idea presented in \cite{vytlacil2006note}.

\begin{proof} [Proof of Theorem \ref{thm:representation}]

The direction $(ii) \implies (i)$ is obvious. For $(i) \implies (ii)$, consider fixing $X=x$ for any $x \in \mathcal{X}$, then apply the results by \cite{vytlacil2002independence}. Formally, let $(\Omega, \mathcal{B},P)$ be the probability space that underlies the random vector $(Y_1,Y_0, ,\{D_z:z \in \mathcal{Z}\},X,Z)$. Consider the partition $\Omega = \bigcup_{x\in\mathcal{X}}\Omega_x$, where $\Omega_x = \{\omega\in\Omega:X(\omega)=x\}$. The probability $P(\Omega_x)$ is non-zero as $\mathcal{X}$ is assumed to be countable. For each $x\in\mathcal{X}$, we construct the probability space $(\Omega_x,\mathcal{B}_x,P_x)$ where the $\sigma$-algebra
\[
\mathcal{B}_x = \left\{B\cap\Omega_x: B\in\mathcal{B}\right\} 
\]
and the probability measure 
\[P_x(B) = \frac{P(B\cap\Omega_x)}{P(\Omega_x)} \;\text{ for }\;B\in\mathcal{B}_x.
\]
Consider the random variable $\left.D_z^x = D_z\right|_{\Omega_x}$, which is the restriction of $D_z$ to the subdomain $\Omega_x$ Similarly define $Y_1^x,Y_0^x,Z^x$. It is not hard to see that the probability space and the random variables are well-defined and that $P_X(B)$ is the conditional probability of $B$ given $X$.\footnote{We are stating that the followings are true: (1) $(\Omega_x,\mathcal{B}_x,P_x)$ is a probability space, (2) $Y_1^x,Y_0^x,Z^x$ and $\{D_z^x:z \in \mathcal{Z}\}$ are $\mathcal{B}_x$-measurable, and (3) $P_X(B) = \mathbb{E}[\mathbf{1}_{B}|X] \text{ a.s. for } B \in \mathcal{B}$.} Then Assumption \ref{ass:CI} implies that for all $z \in \mathcal{Z}$, $Z^x \perp (Y_1^x,Y_0^x,D_z^x)$ under the probability measure $P_x$. Assumption \ref{ass:monotonicity} implies that for all $(z,z') \in \mathcal{Z}^2$, either $D_z^x \geq D_{z'}^x$ or $D_z^x \leq D_{z'}^x$. This means that the variables $(Y_1^x,Y_0^x,\{D_z^x:z \in \mathcal{Z}\},Z^x)$ satisfy the LATE assumptions defined by \cite{vytlacil2002independence} (namely assumptions L-1 and L-2 in that paper). By the result in that paper, an equivalent representation for $D_z^x$ is that there exists a non-trivial function $g$ and a random variable $U^x$, independent of $Z^x$, such that $D^x_z = \mathbf{1}\{ g(z,x) \geqslant U^x \}$. Define $U \equiv U^X $. By construction $U \perp Z \mid X$, and $D_z = D_z^X = \mathbf{1}\{ g(z,X) \geqslant U \}$.

The global monotonicity condition imposes further restrictions. It forbids the following situation: for some pairs $(z,z') \in \mathcal{Z}^2$ and $(x,x') \in \mathcal{X}^2$, we have $ g(z,x) > g(z',x)$ but $g(z,x') < g(z',x')$. This violation is depicted in Figure \ref{fig:violation}, where under $X=x$ we have compliers, but under $X=x'$ we have defiers.

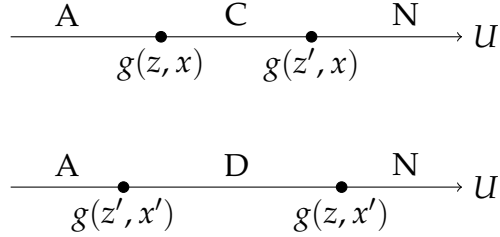
\begin{figure}[h] 
    \centering
    \caption{Violation of Assumption \ref{ass:monotonicity}}
    \label{fig:violation}
    \vspace{10pt}
    \begin{tikzpicture}
		\draw [->] (0,3) -- (6,3) node[ midway, above]{C} node[ very near start, above]{A} node[ very near end, above]{N}  node[right]{$U$};
		\filldraw[black] (2,3) circle (2pt) node[anchor=north]{$g(z,x)$};
		\filldraw[black] (4,3) circle (2pt) node[anchor=north]{$g(z',x)$};
		\draw [->] (0,1) -- (6,1) node[ midway, above]{D} node[ very near start, above]{A} node[ very near end, above]{N} node[right]{$U$};
		\filldraw[black] (1.5,1) circle (2pt) node[anchor=north]{$g(z',x')$};
		\filldraw[black] (4.4,1) circle (2pt) node[anchor=north]{$g(z,x')$};
\end{tikzpicture} \\
\vspace{10pt}
\small A: always taker, C: complier, D: defier, N: never taker.
\end{figure}

In fact, the monotonicity condition implies that $g$ satisfies the following property: for all pairs $(z,z')$, either $g(z,x) > g(z',x)$ for all $x$, or  $g(z,x) < g(z',x)$ for all $x$.
By Lemma 1 in \cite{vytlacil2006note}, this property implies that there exists a set of strictly increasing functions $\{h_x(\cdot):x \in \mathcal{X}\}$ and a function $m$ such that $g(z,x) = h_x(m(z))$. Thus, based on the fact that each $h_x$ is strictly increasing hence invertible, we can derive that
    $D_z = \mathbf{1}\{ g(z,X) \geqslant U \} = \mathbf{1}\{ h_X(m(z)) \geqslant U \} = \mathbf{1}\{ m(z) \geqslant h_X^{-1}(U) \}$ and establish the representation in (\ref{eqn:index_crossing}).\par
	For representation (\ref{eqn:representation_uncond}), define $U^*=F_{U|X}(U \mid X)$ and $q^*(x,u) = q(x,F_{U|X}^{-1}(u \mid x))$, where $F_{U|X}(\cdot \mid \cdot)$ is the conditional cumulative distribution function of $U$ given X.  We have
\begin{align*}
D_z 
&= \mathbf{1}\{ m(z) \geqslant q(X,U) \} \\
&= \mathbf{1}\{ m(z) \geqslant q(X,F_{U|X}^{-1} \circ F_{U|X}(U)) \}\\
&= \mathbf{1}\{ m(z) \geqslant q^*(X,U^*) \}.
\end{align*} 
$U^*$ is independent to $(X,Z)$ and distributed as $\text{Unif}[0,1]$ because
\begin{align*}
P(U^*\leqslant u \mid X=x,Z=z) 
&= P(F_{U|X}(U \mid X)\leqslant u \mid X=x,Z=z)\\
&= P(F_{U|X}(U \mid x)\leqslant u \mid X=x)\\
&= P(U\leqslant F^{-1}_{U|X}(u\mid x) \mid X=x)\\
&= F_{U|X}( F^{-1}_{U|X}(u \mid x) \mid x), \\
& = u, \text{ for all } u \in[0,1],
\end{align*}
where the second line holds as $U \perp Z \mid X$, and the fourth line holds by the definition of $F_{U|X}$.
%The last step applies Lemma 7.11 in \cite{van2000asymptotic}. Let $U$ be a uniform random variable that is independent of $(Z,X)$ and $q(x,\tau)$ be the quantile function of the conditional distribution $\tilde{U} \mid X=x$, i.e. $P(U \leqslant q(x,\tau) \mid X=x) = \tau$. Then we have $ D_z = \mathbf{1}\{ m(z) \geqslant q(X,U) \}  $. %changed%
\end{proof}

\section{Implications} \label{sec:implications}
 
The separability property between $Z$ and $X$ in the choice equation implies a rank-invariance property of the ranking of the instrument in terms of the propensity score 
$$\pi(z,x) \equiv P(D=1 \mid Z=z,X=x).$$
The following corollary also discusses the identification of the function $m$ from the propensity score.

\begin{corollary} [Observable Implications] \label{cor:obs_implication}
    Let Assumptions \ref{ass:CI} and \ref{ass:monotonicity} hold. 
    \begin{enumerate} [label = (\roman*)]
        \item The propensity score $\pi$ satisfies that for any $z,z' \in \mathcal{Z}$,
        \begin{align} \label{eqn:condition_ps}
            \pi(z,x) \geqslant \pi(z',x)  \text{  for some } x \implies \pi(z,x) \geqslant \pi(z',x) \text{  for all } x.
        \end{align}
        Further, using representation (\ref{eqn:representation_without_X}), if the CDF of $\tilde{U}$ is strictly increasing conditional on some value $x^*$, then the function $m$ can be ordinally identified as $m(z) = \pi(z,x^*)$.  Moreover, if the conditional CDF is strictly increasing for all $x\in\mathcal{X}$, then the above statement also holds when the weak inequalities in equation (\ref{eqn:condition_ps}) are replaced by strict inequalities (or equalities).

        \item The function $m$ is a sufficient index of the instrument $Z$ in the sense that
        \begin{align*}
            P(Y_j \in \mathcal{B} \mid X, Z, D=j) = P(Y_j \in \mathcal{B} \mid X,m(Z),D=j),
        \end{align*}
        for any measurable set $\mathcal{B}$ and $j \in \{0,1\}$. Let $g_1,g_0$ be nonnegative functions, then 
        \begin{align*}
            \mathbb{E} \left[ Dg_1(Y,X) \mid X, m(Z) = \mu \right]
        \end{align*}
        is weakly increasing
        in $\mu$ $(w.p.1)$ and 
        \begin{align*}
            \mathbb{E} \left[ (1-D)g_0(Y,X) \mid X, m(Z) = \mu \right]
        \end{align*}
        is weakly decreasing in $\mu$ $(w.p.1)$. These implications are testable when the CDF of $\tilde{U}$ is strictly increasing conditional on some value $x^*$, a case in which the function $m$ can be ordinally identified as the propensity score $\pi(z,x^*)$.
    \end{enumerate}
    
\end{corollary}

This first implication means that the function $m$ provides an observable ordering of the instrument values by their strength of pushing individuals to take up the treatment. That is, in the CLATE model, we can rank the instrument values by their effectiveness of inducing individuals into the treatment status. This ordering remains invariant under different values of $X$ because the monotonicity is assumed to be global (Assumption \ref{ass:monotonicity}). However, we do not impose the ``normalization'' that a higher value of the instrument always leads to more treatment take-ups, so the function $m$ need not be increasing. 

The second implication uses the identified $m$ to derive a set of testable implications of the CLATE model. This set of testable implications is a refinement of the testable implications of the marginal treatment effect framework derived in \cite{heckman2005structural} as the role of Z is fully summarized by the function $m$.\footnote{Notice that notation ``$Z$'' in \cite{heckman2005structural} is different from ours as it represents the joint set of the instruments and covariates. By contrast, $Z$ only contains the excluded instruments in our paper. Accordingly, the set of testable implications we derive is also stronger in that $m$ only depends on the excluded instrument, which is a result of the monotonicity condition imposed by the CLATE model.} The testable implications are also analogous to those presented in Equation (3.3) in \cite{kitagawa2015test} except that, here, $m(Z)$, but not $Z$ itself, enters the conditioning set. That is, the testable implications we derived do not restrict the direction of the effect of $Z$ on the treatment take-up. The distinction appears as we do not explicitly assume that no defier exists as \cite{kitagawa2015test} does. We only assume that defiers and compliers can not both exist. That is, as stated in Assumption \ref{ass:monotonicity}, we leave the direction of monotonicity unspecified. \footnote{If the function $m$ were known and increasing, the testable implication in this paper would essentially reduce to Equation (3.3) in \cite{kitagawa2015test}, except that $Z$ can possibly be non-binary in our case.  Combined with the testing procedure proposed in Section 3.1 in his paper to handle multivalued instruments, we may likewise design a test for our implication.}

\begin{proof} [Proof of Corollary \ref{cor:obs_implication}]
    \begin{enumerate} [label = (\roman*)]
        \item From (\ref{eqn:representation_without_X}), %changed%
        we have $\pi(Z,X) = P(m(Z) \geqslant \tilde{U} \mid Z,X) = F_{\tilde{U} \mid X}(m(Z))$, where $F_{\tilde{U}\mid X}$ denotes the conditional CDF of $\tilde{U}$ given $X$. The Condition (\ref{eqn:condition_ps}) on the propensity score is satisfied since $F_{\tilde{U} \mid X}$ is non-decreasing. When $F_{\tilde{U} \mid X=x^*}$ is strictly increasing, the ordinal information contained in $m(\cdot)$ is fully transformed into $\pi(\cdot,x^*)$.

        Now suppose that $F_{\tilde{U} \mid X=x^*}$ is strictly increasing for all $x$. As a result, the function $m$ is ordinally identified by $\pi(z,x)$ for any $x\in\mathcal{X}$. By the definition of being ordinally identified, we have 
	\[
	\pi(z,x) > \pi(z',x) \iff m(z) > m(z') \iff \pi(z,x') > \pi(z',x')
	\]
	and 
	\[
	\pi(z,x) = \pi(z',x) \iff m(z) = m(z') \iff \pi(z,x') = \pi(z',x')
	\]
for any $z,z'\in\mathcal{Z}$ and $x,x'\in\mathcal{X}$.

        \item Consider the case where $j=1$, the other case can be proved by symmetric arguments. By representation (\ref{eqn:index_crossing}), we have
        \begin{align*}
	 P(Y_1 \in \mathcal{B} \mid X=x, Z=z, D=1) & = P(Y_1 \in \mathcal{B} \mid X=x, Z=z, m(z)\geqslant q(x,U)) \\
            & = P(Y_1 \in \mathcal{B} \mid X=x,  m(z)\geqslant q(x,U)) \\
            & = P(Y_1 \in \mathcal{B} \mid X=x, m(Z)=m(z) ,m(z)\geqslant q(x,U)) \\
            & = P(Y_1 \in \mathcal{B} \mid X=x, m(Z)=m(z), m(Z)\geqslant q(X,U)) \\
            & = P(Y_1 \in \mathcal{B} \mid X=x, m(Z)=m(z), D=1),           
        \end{align*}
        where the second and third lines follow from the conditional independence assumption (Assumption \ref{ass:CI}), and the last line follow from the equivalence result.\footnote{We thank one of the anonymous referee for suggesting this proof to improve the clarity of the original arguments.} For the second assertion, let $\mu > \mu'$, then
        \begin{align*}
            &\mathbb{E} \left[ Dg_1(Y,X) \mid X, m(Z) = \mu \right] - \mathbb{E} \left[ Dg_1(Y,X) \mid X, m(Z) = \mu' \right] \\
            = & \mathbb{E} \left[ \mathbf{1}\{h(X,U) \leqslant \mu\}g_1(Y_1,X) \mid X, m(Z) = \mu \right] \\
            -& \mathbb{E} \left[ \mathbf{1}\{h(X,U) \leqslant \mu'\}g_1(Y_1,X) \mid X, m(Z) = \mu' \right] \\
            = & \mathbb{E} \left[ \mathbf{1}\{ \mu' < h(X,U) \leqslant \mu\}g_1(Y_1,X) \mid X \right] \geqslant 0.
        \end{align*}
        The case of $\mathbb{E} \left[ (1-D)g_0(Y,X) \mid X, m(Z) = \mu \right]$ is similar. When the function $m$ is ordinally identified as $m(z) = \pi(z,x^*)$ for some $x^*$, we can rewrite the implications as (i) the following equality
        \[P(Y_j \in \mathcal{B} \mid X, Z, D=j) = P(Y_j \in \mathcal{B} \mid X,\pi(Z,x^*),D=j),
        \]
	hold for any measurable set $\mathcal{B}$ and $j \in \{0,1\}$, (ii) the function
        \[       
         \mathbb{E} \left[ Dg_1(Y,X) \mid X, \pi(Z,x^*) = p \right]
         \] is weakly increasing for $p$ in the range of $\pi(Z,x^*)$, and (iii) the function
         \[
         \mathbb{E} \left[ (1-D)g_1(Y,X) \mid X, \pi(Z,x^*) = p \right]
         \]
         is weakly decreasing for $p$ in the range of $\pi(Z,x^*)$. These implications are testable as $\pi(Z,x^*)$ is identified by the observed propensity $P(D=1 \mid Z=z,X=x^*)$.
    \end{enumerate}
    
\end{proof}

\section{Ordered Treatment Levels} \label{sec:ordered}

This section extends the representation result in Section \ref{sec:main} to incorporate multiple ordered levels of treatment. The argument follows from the equivalence results in \cite{vytlacil2006ordered}. Let there be $K$ possible levels of treatment. Now the treatment $D$ takes values in an ordered set $\{1,\cdots,K\}$. The counterfactual treatment $D_z$'s are defined accordingly. The corresponding potential outcomes are denoted by $(Y_1,\cdots,Y_K)$. 

The CLATE assumptions are modified to incorporate the ordered multiplicity in treatment levels. Although we have a different definition of $D$, the statement of the monotonicity condition does not change.

\begin{manualtheorem}{\ref{ass:CI}'}\label{order_CI}
    $\left( \{D_z : z \in \mathcal{Z}\}, Y_1, \cdots, Y_K \right) \perp Z \mid X $.
\end{manualtheorem}

\begin{manualtheorem}{\ref{ass:monotonicity}'}\label{order_monotonicity}
    For any $(z,z') \in \mathcal{Z}^2$, either 
    $$ P(D_z \geqslant D_{z'} \mid X=x) =1 \text{, for almost all } x$$
    or 
    $$ P(D_z \leqslant D_{z'} \mid X=x) =1 \text{, for almost all } x.$$
\end{manualtheorem}

\begin{corollary} [Ordered Treatment Levels] \label{cor:order}
    The ordered CLATE model (Assumptions \ref{order_CI} and \ref{order_monotonicity}) is equivalent to the following statements. There exist a function $m$ and $K+1$ random variables $U_0, \cdots,U_K$ such that for $k=1,\cdots,K$,
    \begin{enumerate} [label = (\roman*)]
        \item $D_z = k \iff U_{k-1} \leqslant m(z) < U_{k}$,
        \item $Z \perp \left( U_1, \cdots,U_{K-1}, Y_1, \cdots, Y_{K} \right) \mid X $, 
        \item $U_0 = -\infty$, $U_K = \infty$, and $U_k \geqslant U_{k-1}$.
    \end{enumerate}
\end{corollary}

This is basically the conditional version of the representation result in \cite{vytlacil2006ordered}. The main point is that even though the random thresholds $U_1,\cdots,U_K$ covariates with $X$, the
latent index $m(Z)$ does not explicitly depend on $X$. Again, this is because Assumption \ref{order_monotonicity} requires that the direction of monotonicity has to be the same across all values of $X$. 

\begin{proof} [Proof of Corollary \ref{cor:order}]
    Proof of direction from the random thresholds model to the CLATE model can be done in the exact same way as in \cite{vytlacil2006ordered}. In particular, Assumption \ref{ass:CI}' is implied by item (ii). For Assumption \ref{ass:monotonicity}', suppose for some $(z,z') \in \mathcal{Z}^2$ and $x \in \mathcal{X}$, $P(D_z > D_{z'} \mid X) > 0$. This implies that $m(z) > m(z')$, which in turn implies that $D_z \geq D_{z'}$. This in fact means that $P(D_z \geq D_{z'} \mid X=x) =1$ for all $x \in \mathcal{X}$. Because otherwise, $P(D_z \geq D_{z'} \mid X=x) <1$ for some $x \in \mathcal{X}$, then by the law of total probability, 
    $$P(D_z \geq D_{z'}) = \sum_{x \in \mathcal{X}} P(X=x) P(D_z \geq D_{z'} \mid X=x)<1.$$
    
    For the other direction, define $D_z^k = \mathbf{1}\{D_z > k\}$. Then each $D_z^K$ is a binary treatment whose representation can be analyzed by Theorem \ref{thm:representation}. So $D_z^k = \mathbf{1}\{m^k(z) \geqslant \tilde{U}^k\}$. For any $z,z' \in \mathcal{Z}$, $m^k(z) \geqslant m^k(z')$ implies that $D_z \geqslant D_{z'}$ by monotonicity. Let $d(z) = \mathbb{E}\left[ D_z \right] $. Then by Lemma 1 in \cite{vytlacil2006ordered}, $m^k(z) = g^k(d(z))$ for some non-decreasing $g^k(\cdot)$. The rest of the proof follows from that paper.
\end{proof}

\section{Conclusion}

This paper shows that the CLATE model has a latent index representation in which the instrument and the covariates are separable in the treatment choice equation. On the theoretical side, the result more rigorously links the CLATE model to the latent index representation when covariates are present. On the practical side, the result establishes conditions when methods from the two pieces of literature can be used interchangeably. For example, one can employ the nonparametric estimator in \cite{frolich2007nonparametric} as robustness checks for the structural estimates in selection models. For future works, one can consider extending this result into the unordered monotonicity model \citep{heckman2018unordered}.

\bibliographystyle{apalike}
\bibliography{references.bib}

\end{document}